\documentclass[11pt]{article}

\usepackage[a4paper,margin=2.5cm]{geometry}

\usepackage[T1]{fontenc}
\usepackage[utf8]{inputenc}
\usepackage{lmodern}
\usepackage{amsmath,amssymb,amsthm,mathtools}
\usepackage{microtype}
\usepackage{xcolor}
\usepackage{hyperref}

\usepackage{authblk}
\usepackage{setspace}

\numberwithin{equation}{section}

\newtheorem{theorem}{Theorem}[section]
\newtheorem{proposition}[theorem]{Proposition}

\theoremstyle{definition}

\theoremstyle{remark}

\hypersetup{
  colorlinks=true,
  linkcolor=blue!45!black,
  citecolor=green!45!black,
  urlcolor=blue!55!black
}

\title{\Large\bfseries A-type Sigma Models from Differential Poisson Geometry}

\author{Cesar Arias${}^\heartsuit$}
\author{Per Sundell${}^\diamondsuit$}

\affil{${}^\heartsuit$Departamento de Matemática\protect\\ 
Pontificia Universidad Católica de Chile, Santiago, Chile \protect\\ 
\texttt{cesar.arias@uc.cl}}

\affil{${}^\diamondsuit$Instituto de Ciencias Exactas y Naturales, Facultad de Ciencias\protect\\ 
Universidad Arturo Prat, Iquique, Chile\protect\\
\texttt{per.anders.sundell@gmail.com}}

\date{}

\begin{document}
\setstretch{1.1}
\maketitle

\begin{abstract}
\noindent We study the differential Poisson sigma model (DPSM) in the symplectic case and show that its classical reduction defines a distinguished class of A-type models on symplectic targets, not necessarily K\"ahler.
The DPSM is a covariant first-order sigma model whose graded target is the parity-shifted tangent bundle $T[1]M$ of a Poisson manifold $M$. Its graded Poisson tensor encodes a differential Poisson bracket on $C(T[1]M)\cong\Omega^\bullet(M)$, written covariantly in terms of a connection $\Gamma$ and its transpose $\widetilde\Gamma$.
In the nondegenerate case, the Jacobi identities force $\Gamma$ to be flat, while the quartic coupling of the reduced action is given by the curvature of $\widetilde\Gamma$, induced by the torsion of $\Gamma$. 
Thus, the DPSM selects a symplectic class in which the A-model curvature coupling acquires a first-order Poisson origin. 
We describe this class through examples and obstructions; $\mathbb{CP}^n$ and K3 surfaces are excluded, while affine symplectic targets, symplectic tori, and the Kodaira--Thurston manifold furnish explicit examples.
The graded parent geometry on $T[1]M$ equips $\Omega^\bullet(M)$ with a differential graded Poisson algebra structure; in particular, the underlying differential graded Lie algebra defines a strict $L_\infty$-algebra on the observable complex. This chain-level structure is not manifest in the usual K\"ahler formulation of the A-model.
\end{abstract}

\newpage
\tableofcontents
\newpage
\section{Introduction}
\label{section1}
The topological A-model is usually introduced as the A-twist of the two-dimensional $N=(2,2)$ supersymmetric sigma model~\cite{Witten1988,WittenMirror1992}. In that formulation, off-shell closure of the supersymmetry algebra requires the target to be K\"ahler, and the classical A-model action contains the familiar topological term, fermionic kinetic term, and quartic curvature coupling. Quantum mechanically, the theory localizes on pseudo-holomorphic maps, providing a particularly tractable instance of  Gromov--Witten theory and quantum cohomology~\cite{KontsevichManin1994, Dubrovin:1994hc}.

At the classical level, the space of local observables of the A-model contains worldsheet zero-forms with a simple description in terms of the de Rham complex of the target. Let $X:\Sigma\to M$ be the sigma-model map, and let $(x^a,\theta^a)$ denote the induced local coordinates on the graded target $T[1]M$, with $\theta^a$ the odd fiber coordinates transforming as components of a section of $X^*TM[1]$. 
Under the isomorphism 
\begin{equation}
\label{isomorphism}
C(T[1]M)\cong\Omega^\bullet(M),
\end{equation}
where all tensorial spaces over $M$ are assumed to be smooth, the de Rham differential is intertwined with the canonical cohomological vector field
\begin{equation}
Q=\theta^a\frac{\partial}{\partial x^a}
\end{equation}
on $T[1]M$. Thus, if $\eta\in\Omega^\bullet(M)$ corresponds to $F_\eta\in C(T[1]M)$, then
\begin{equation}
QF_\eta=F_{d\eta}.
\end{equation}
Evaluating $F_\eta$ on the worldsheet fields defines a local zero-form observable $\mathcal O_\eta$. The field-space cohomological symmetry
$\mathcal Q$ acts on such observables by pullback of $Q$, and hence
\begin{equation}
\mathcal Q\mathcal O_\eta=\mathcal O_{d\eta}.
\end{equation}
This does not exhaust the full space of local operators of the theory, since one
may also consider operators involving other fields or worldsheet derivatives, but it isolates the natural class of classical cohomological observables considered here.
On a closed worldsheet, the corresponding classical $\mathcal Q$-cohomology is therefore $H^\bullet_{\mathrm{dR}}(M)$.

The purpose of this paper is to show that this classical A-model package arises upon symplectic reduction of differential Poisson sigma models (DPSM), but in a form not tied to the usual K\"ahler presentation. The symplectic case of the DPSM selects a distinguished class of symplectic targets equipped with flat parent connection data. The reduced classical actions define A-type models whose curvature couplings are not introduced as Levi--Civita curvatures of K\"ahler metrics, but arise from the curvature of a Poisson-compatible transposed connection $\widetilde\Gamma$.

A useful point of comparison is the relation between the ordinary Poisson sigma model (PSM)  \cite{Ikeda1994,SchallerStrobl1994} and the A-model. A PSM has target given by a Poisson manifold $(M,\pi)$ and first-order action
\begin{equation}
S_{\mathrm{PSM}}=\int_\Sigma\left[ p_a\wedge dx^a+\frac12\,\pi^{ab}(x)p_a\wedge p_b\right],
\end{equation}
where $x^a$ are local coordinates on $M$ pulled back to $\Sigma$, and $p_a$ are worldsheet one-forms valued in $X^*T^*M$. In the AKSZ formulation~\cite{Alexandrov:1995kv,Cattaneo:1999fm}, the classical master equation is equivalent to the Jacobi identity for $\pi$. When $\pi$ is nondegenerate, its inverse $\omega=\pi^{-1}$ is a symplectic form, and suitable BV gauge fixings relate the symplectic PSM to the A-model~\cite{Bonechi:2007ar,Bonechi:2016wqz}. 

There also exist extensions of topological sigma models beyond the ordinary K\"ahler setting. Topological sigma models with $H$-flux lead to A- and B-type theories naturally formulated in terms of twisted generalized complex geometry~\cite{Kapustin:2004gv}. The worldsheet origin of generalized complex geometry in supersymmetric sigma models was also developed in~\cite{Lindstrom:2004iw}, and the corresponding D-brane geometry was studied in~\cite{Zabzine:2004dp}. 
Topological strings on noncommutative manifolds provide another example, where deformations of two-dimensional $N=(2,2)$ sigma models are described using generalized complex structures~\cite{Kapustin:2003sg}. These constructions show that topological sigma models need not be confined to ordinary K\"ahler geometry. 
The mechanism developed in the present paper is different, since it does not start from an $H$-flux, a noncommutative deformation, or a generalized complex structure, but from the DPSM and the symplectic reduction of its graded Poisson target.

The DPSM~\cite{Arias:2015wha,Arias:2016agc} is a graded extension of the PSM, closely related to the constructions of~\cite{Zucchini:2004ta,Ikeda:2007rn} and to subsequent developments~\cite{Basile:2025kib}. Here we use the original DPSM structure to study its nondegenerate symplectic phase and the A-type theory obtained by reduction.
Related semiclassical differential structures on forms, including examples arising from noncommutative tori and Poisson--Lie groups, were studied in~\cite{Beggs:2003pm}.

The graded target of the DPSM is the parity-shifted tangent bundle $T[1]M$, so that the isomorphism~\eqref{isomorphism} holds. A graded Poisson tensor $\Pi$ on $T[1]M$ therefore defines a bracket on differential forms. Requiring this bracket to be written covariantly introduces a connection on $M$, and compatibility with the canonical cohomological vector field on $T[1]M$ fixes the mixed and odd components of $\Pi$. The resulting first-order theory has fields $(x^a,\theta^a;p_a,\chi_a)$; the pair $(x^a,\theta^a)$ are coordinates on $T[1]M$ pulled back to the worldsheet, while $p_a$ and $\chi_a$ are their conjugate one-form fields. In this sense the DPSM is a covariant doubling of the ordinary PSM, with a quartic fermionic term already present at the parent level.

A central feature of the covariant formulation of the DPSM is the appearance of two transposed connections. We denote them by $\Gamma$ and $\widetilde\Gamma$, related in a local coordinate basis by
\begin{equation}
\widetilde\Gamma^a{}_{bc}=\Gamma^a{}_{cb}\ .
\end{equation}
In the differential Poisson bracket of forms, $\Gamma$ defines the covariant derivatives, whereas $\widetilde\Gamma$ enters the graded Poisson tensor itself. The same separation is reflected in the DPSM action, in which $\Gamma$ governs the fermionic kinetic term, while $\widetilde\Gamma$ controls the cubic and quartic couplings.
The Jacobi identities of the differential Poisson bracket, or, equivalently, the gauge invariance of the DPSM, impose strong restrictions on this pair. In the symplectic case, they force $\Gamma$ to be flat while it leaves room for torsion.
Thus the curvature inherited by the reduced theory is not the ordinary Levi--Civita curvature, but rather $\widetilde R$, which is induced by the torsion.

In the symplectic case, under the assumption that $\widetilde\nabla\pi=0$, elimination of the one-form momenta $p_a$ yields the effective DPSM action
\begin{equation}
S_{\mathrm{red}}=\int_\Sigma
\left[ \frac12\,\omega_{ab}\,dx^a\wedge dx^b+\chi_a\wedge\nabla\theta^a 
+\frac14\,\pi^{de}\widetilde R_{ab}{}^{c}{}_{e} \theta^a\theta^b\chi_c\wedge\chi_d \right]\ ,
\end{equation}
defining a subclass of A-type models that need not be K\"ahler but that are constrained by differential Poisson compatibility conditions. In particular, the flatness of the parent
connection $\Gamma$ gives a simple obstruction in the simply connected case. 
If $M$ is simply connected and compact, and admits a flat parent connection on $TM$, then parallel transport has trivial holonomy and $TM$ is globally trivializable. Thus, simply connected targets with nontrivial tangent bundle, such as $\mathbb{CP}^n$ and K3 surfaces, are excluded. 
At the same time, the class contains natural symplectic examples, including affine symplectic targets, symplectic tori, complex tori with translation-invariant K\"ahler form, and symplectic nilmanifolds~\cite{BensonGordon1988, Hasegawa1989}. The Kodaira--Thurston manifold~\cite{Kodaira1963, Thurston1976} provides a compact non-K\"ahler example with nonzero transposed curvature.

The local observable sector is inherited directly from the graded target. Differential forms on $M$ define local observables, and the cohomological vector field acts as the de Rham differential. In addition, the graded Poisson tensor $\Pi$ equips $\Omega^\bullet(M)$ with a degree-zero differential Poisson bracket. Together with the wedge product and the de Rham differential, this makes \(\Omega^\bullet(M)\) into a differential graded Poisson algebra. In particular, the differential $d$ and the bracket $\{\cdot,\cdot\}$ define a strict $L_\infty$-algebra on the same observable complex. Thus the DPSM does not only produce the reduced action; it also identifies a chain-level algebraic structure behind the classical observable complex.

We emphasize that the analysis in this paper is classical. The reduced theory reproduces the classical data of the A-model---a topological symplectic term, a fermionic kinetic term, a quartic curvature coupling, and a cohomological symmetry acting as the de Rham differential---but we do not treat its quantum localization on pseudo-holomorphic maps, quantum cohomology, or Gromov--Witten theory. These lie beyond the classical scope considered here, and it is in this classical sense that we refer throughout to A-type models.

The paper is organized as follows. 
Section~\ref{section2} reviews the covariant formulation of the DPSM, including the graded target $T[1]M$, the differential Poisson tensor, the pair of transposed connections, and the resulting first-order action. It also includes a nonsymplectic $T^3$ example, illustrating that the parent DPSM is defined more broadly than the symplectic phase used later in the reduction. 
Section~\ref{section3} performs the symplectic reduction and derives the corresponding classical A-type action. 
Section~\ref{section4} discusses the classical observable complex, its differential graded Poisson algebra structure, and the associated strict $L_\infty$-algebra.
Section~\ref{section5} gives examples and obstructions, including affine symplectic targets, the simply connected obstruction, and the Kodaira--Thurston non-K\"ahler background with nonzero transposed curvature. 
Section~\ref{section6} concludes with comments on possible deformations of the differential Poisson algebra and their potential relation to quantum A-model structures.

\section{Differential Poisson sigma models}
\label{section2}
In this section we review the DPSM as a covariant graded extension of the ordinary PSM. The central idea is to enlarge the target from a Poisson manifold $M$ to the graded manifold $T[1]M$, so that differential forms on $M$ are incorporated directly into the target-space geometry~\cite{Arias:2015wha,Arias:2016agc}. 
The resulting graded Poisson structure defines a Poisson bracket for differential
forms involving a pair of transposed connections $\Gamma$ and $\widetilde\Gamma$.
In the nondegenerate case, the Jacobi identities force $\Gamma$ to be flat,
while $\widetilde\Gamma$ may still have nonzero curvature.

\subsection{Graded target and differential Poisson bracket}
Let $M$ be a smooth manifold and consider the graded manifold $T[1]M$, with local coordinates
\begin{equation}
Z^A=(x^a,\theta^a), \qquad \deg x^a=0, \qquad \deg\theta^a=1.
\end{equation}
Here $\deg$ denotes the internal degree on the graded target $T[1]M$. The odd variables $\theta^a$ transform as components of a tangent vector under changes of coordinates on $M$. Hence, functions on $T[1]M$ are naturally identified with differential forms on $M$
\begin{equation}
\label{isom}
C(T[1]M)\cong \Omega^\bullet(M), \qquad
\theta^a\mapsto dx^a .
\end{equation}
Under this identification, the de Rham differential is represented by the canonical cohomological vector field
\begin{equation}
Q=\theta^a\frac{\partial}{\partial x^a}, \qquad Q^2=0.
\end{equation}

A differential Poisson bracket on $\Omega^\bullet(M)$ is encoded by a graded Poisson tensor 
\begin{equation}
\Pi=\frac12\,\Pi^{AB}(x,\theta)
\frac{\partial}{\partial Z^A}\wedge
\frac{\partial}{\partial Z^B}
\end{equation}
on $T[1]M$. The induced bracket 
\begin{equation}
\{F, G\}=\Pi^{AB}\,\partial_A F\,\partial_B G
\end{equation}
on $C(T[1]M)$ is graded antisymmetric and has degree zero, namely
\begin{equation}
\{F, G\} = -(-1)^{\deg F \deg G}\{G, F\},\qquad
\deg\{F, G\}=\deg F+\deg G.
\end{equation}
The graded Jacobi identity is equivalent to
\begin{equation}
\label{Schouten}
[\Pi,\Pi]=0,
\end{equation}
where $[\cdot,\cdot]$ is the graded Schouten bracket. Compatibility with the de Rham differential is the condition that $Q$ acts as a derivation of the bracket, that is
\begin{equation}
\label{Qder}
Q\{F,G\} = \{QF,G\} + (-1)^{\deg F}\{F,QG\}.
\end{equation}
%

\subsection{Covariant doubling of the PSM and the two connections}
The worldsheet fields of the DPSM are a covariant doubling of the fields of the ordinary PSM. Besides the map $X:\Sigma\to M$, with local components $x^a$, and its one-form momenta
\begin{equation}
p_a\in\Omega^1(\Sigma,X^*T^*M),
\end{equation}
we introduce
\begin{equation}
\theta^a\in\Omega^0(\Sigma,X^*TM),\qquad
\chi_a\in\Omega^1(\Sigma,X^*T^*M).
\end{equation}
When the target coordinates are pulled back to the worldsheet, we distinguish the internal degree from the worldsheet form degree. The total degree is the sum of these two degrees. With this convention
\begin{equation}
\deg_{\mathrm{tot}} x^a=0,\qquad \deg_{\mathrm{tot}} p_a=1,\qquad
\deg_{\mathrm{tot}}\theta^a=1,\qquad \deg_{\mathrm{tot}}\chi_a=2.
\end{equation}

The covariant formulation involves two connections related by transposition of their lower indices. We denote them by $\Gamma$ and $\widetilde{\Gamma}$, with
\begin{equation}
\widetilde{\Gamma}^a{}_{bc}:=\Gamma^a{}_{cb}.
\end{equation}
In a coordinate frame this may also be written as
\begin{equation}
\widetilde{\Gamma}^a{}_{bc}=
\Gamma^a{}_{bc}-T^a{}_{bc},\qquad T^a{}_{bc}:=2\Gamma^a{}_{[bc]},
\end{equation}
where $T$ is the torsion of $\Gamma$.

The connection $\Gamma$ is used to build a covariant fermionic kinetic term. While the term $p_a\wedge dx^a$ is already globally defined, its fermionic analogue $\chi_a\wedge d\theta^a$ is not covariant on a curved target, since $\theta^a$ transforms as a vector. The covariant derivative is
\begin{equation}
\nabla\theta^a = d\theta^a + dx^b\,\Gamma^a{}_{bc}\,\theta^c.
\end{equation}
Correspondingly, one performs the connection-dependent shift
\begin{equation}
\label{lambda}
p_a \rightarrow \lambda_a:= p_a-\Gamma^c{}_{ab}\,\theta^b\chi_c,
\end{equation}
so that the graded coordinates and covariantized momenta
\begin{equation}
\label{ZP}
Z^A = (x^a, \theta^a), \qquad P_A=(\lambda_a,  \chi_a),
\end{equation}
combine to give
\begin{equation}
P_A\wedge dZ^A 
= p_a\wedge dx^a +\chi_a\wedge \nabla\theta^a.
\end{equation}
This graded kinetic term is a globally defined scalar two-form on $\Sigma$.

The transposed connection $\widetilde{\Gamma}$ is taken to be compatible with the Poisson tensor
\begin{equation}
\widetilde{\nabla}_c\pi^{ab}=0.
\end{equation}
The above distinction between $\Gamma$ and $\widetilde{\Gamma}$ is essential. The kinetic term uses $\Gamma$, while the Poisson compatibility condition uses $\widetilde{\Gamma}$.

\subsection{Graded Poisson tensor}
We now determine the graded Poisson tensor entering the DPSM. We begin with brackets compatible with the grading and whose bosonic part reproduces the ordinary Poisson structure, that is
\begin{equation}
\{x^a,x^b\}=\pi^{ab}(x),
\qquad
\{x^a,\theta^b\}=M^{ab}{}_c(x)\theta^c,
\qquad
\{\theta^a,\theta^b\}=N^{ab}{}_{cd}(x)\theta^c\theta^d .
\end{equation}
Applying the compatibility condition~\eqref{Qder} to the bracket $\{x^a, x^b\}$ and using the action $Qx^a=\theta^a$, gives
\begin{equation}
Q\{x^a,x^b\}=\{\theta^a,x^b\}+\{x^a,\theta^b\}.
\end{equation}
Since $\{x^a,x^b\}=\pi^{ab}$, this implies
\begin{equation}
M^{ab}{}_c-M^{ba}{}_c=\partial_c\pi^{ab}.
\end{equation}
This condition fixes only the antisymmetric part of $M^{ab}{}_c$. The symmetric part is not determined by compatibility with $Q$. Therefore, using the Poisson-compatible connection $\widetilde{\Gamma}$, we may write the most general covariant representative in the form
\begin{equation}
\label{bracket01S}
\{x^a,\theta^b\} =-\left( \pi^{ad}\widetilde{\Gamma}^b{}_{dc} +S^{ab}{}_c \right)\theta^c,
\end{equation}
where $S^{ab}{}_c=S^{ba}{}_c$ is a tensorial symmetric remainder. 

The tensor $S$ is not invariant data in the symplectic phase. It reflects the freedom in the choice of Poisson-compatible connection used to write the mixed bracket. Indeed, suppose we replace the compatible connection by another compatible connection
\begin{equation}
\widetilde{\Gamma}^a{}_{bc}\rightarrow \widetilde{\Gamma}^a{}_{bc} +\delta\widetilde{\Gamma}^a{}_{bc}.
\end{equation}
Thus, for the mixed bracket~\eqref{bracket01S} to be kept fixed, the shift of $\widetilde{\Gamma}$ must be compensated by the shift 
\begin{equation}
S^{ab}{}_c \rightarrow S^{ab}{}_c-\pi^{(a|d|}\delta\widetilde{\Gamma}^{b)}_{dc}. 
\end{equation}
Hence, when $\pi$ is nondegenerate, any symmetric $S^{ab}{}_c$ can be absorbed into the choice of compatible connection. In the symplectic phase we shall therefore work in the gauge
\begin{equation}
S^{ab}{}_c=0.
\end{equation}
The mixed bracket then becomes
\begin{equation}
\label{bracket01}
\{x^a,\theta^b\}=-\pi^{ad}\widetilde{\Gamma}^b{}_{dc}\theta^c =-\pi^{ad}\Gamma^b{}_{cd}\theta^c.
\end{equation}

The odd-odd bracket is fixed by applying~\eqref{Qder} to $\{x^a, \theta^b\}$. Since $Q\theta^b=0$, one obtains
\begin{equation}
\{\theta^a,\theta^b\} = Q\{x^a,\theta^b\}.
\end{equation}
Using~\eqref{bracket01} to compute the right hand side and rewriting the result in terms of the curvature of $\widetilde{\Gamma}$
\begin{equation}
\widetilde R_{ab}{}^c{}_d
=2\partial_{[a}\widetilde{\Gamma}^c{}_{b]d}
+2\widetilde{\Gamma}^c{}_{e[a}\widetilde{\Gamma}^e{}_{b]d},
\end{equation}
one finds
\begin{equation}
\{\theta^a,\theta^b\} =\left(
-\tfrac12\,\widetilde R_{cd}{}^{ab}+\pi^{ef}\Gamma^a{}_{ec}\Gamma^b{}_{fd} 
\right)\theta^c\theta^d,
\end{equation}
where $\widetilde R_{ab}{}^{cd}:=\pi^{de}\widetilde R_{ab}{}^{c}{}_e$. Thus, in the gauge $S=0$, the graded Poisson tensor is
\begin{equation}
\label{gradedP}
\Pi^{AB}(x,\theta)
=
\begin{pmatrix}
\pi^{ab}
&
-\pi^{ad}\Gamma^b{}_{cd}\theta^c
\\[4pt]
\pi^{bd}\Gamma^a{}_{cd}\theta^c
&
\left(
-\frac12\,\widetilde R_{cd}{}^{ab}
+
\pi^{ef}\Gamma^a{}_{ec}\Gamma^b{}_{fd}
\right)\theta^c\theta^d
\end{pmatrix}.
\end{equation}

\subsection{DPSM action}
The DPSM action is the universal first-order sigma-model action on the graded target
\begin{equation}
S_{\mathrm{DPSM}} = \int_\Sigma\left[ P_A\wedge dZ^A + \frac12\,\Pi^{AB}(Z)P_A\wedge P_B \right].
\end{equation}
Substituting the graded coordinates and shifted momenta~\eqref{ZP}, and the graded Poisson tensor~\eqref{gradedP} into this expression, the connection-square terms cancel and one obtains the covariant component action
\begin{equation}
\label{DPSMaction}
S_{\mathrm{DPSM}}= \int_\Sigma \left[
p_a\wedge dx^a + \chi_a\wedge\nabla\theta^a + \frac12\,\pi^{ab}p_a\wedge p_b
+ \frac14\,\widetilde R_{ab}{}^{cd}\theta^a\theta^b\chi_c\wedge\chi_d \right],
\end{equation}
where the fermionic kinetic term uses $\Gamma$, while the quartic coupling is governed by the curvature of the transposed Poisson-compatible connection $\widetilde{\Gamma}$.

\subsection{Parent Jacobi identities and flat parent connection}
The remaining condition on the graded Poisson tensor is the graded Jacobi identity~\eqref{Schouten}. In the covariant formulation, this gives the DPSM target-space identities~\cite{Arias:2015wha}
\begin{equation}
\label{Jacobi}
\pi^{d[a}T^b{}_{de}\pi^{c]e}=0,
\qquad
\pi^{ae}\pi^{bf}R_{ab}{}^c{}_d=0,
\qquad
\pi^{ab}\nabla_b\widetilde R_{cd}{}^{ef}=0,
\end{equation}
and
\begin{equation}
\label{JacobiRR}
\widetilde R_{a[b}{}^{(cd}\widetilde R_{mn]}{}^{p)q}=0.
\end{equation}
In the above, $R$ is the curvature of $\Gamma$ whereas $\widetilde R$ is the curvature of $\widetilde{\Gamma}$. 
The first two identities involve the torsion and curvature of the parent connection $\Gamma$. The last two identities constrain the curvature of the transposed connection that appears in the quartic term.

The second identity in~\eqref{Jacobi} has an immediate consequence in the symplectic phase. If $\pi$ is nondegenerate, it implies
\begin{equation}
R_{ab}{}^c{}_d=0.
\end{equation}
Thus the parent connection $\Gamma$ is flat. Note that this does not imply that the quartic coupling vanishes. Since $\widetilde{\Gamma}$ differs from $\Gamma$ by torsion, the two curvatures are related by
\begin{equation}
\widetilde R_{ab}{}^c{}_d = R_{ab}{}^c{}_d - 2\nabla_{[a}T^c{}_{b]d} + 2T^c{}_{[a|e|}T^e{}_{b]d}.
\end{equation}
Combining this with $R_{ab}{}^c{}_d=0$, one obtains in the strict symplectic DPSM phase
\begin{equation}
\label{RT}
\widetilde R_{ab}{}^c{}_d = -2\nabla_{[a}T^c{}_{b]d} + 2T^c{}_{[a|e|}T^e{}_{b]d}.
\end{equation}
Therefore, in the symplectic case, the quartic coupling in~\eqref{DPSMaction} is determined by the torsion of the flat connection $\Gamma$.

\subsection{Cohomological symmetry and field equations}
The cohomological vector field on $T[1]M$ acts on the shifted variables by
\begin{equation}
\label{Qaction}
Qx^a=\theta^a, \qquad Q\theta^a=0, \qquad Q\lambda_a=0, \qquad Q\chi_a=-\lambda_a.
\end{equation}
Thus the fields split into two nilpotent doublets
\begin{equation}
(x^a,\theta^a), \qquad (\chi_a,-\lambda_a),
\end{equation}
and
\begin{equation}
Q^2=0
\end{equation}
off shell in the first-order theory. In terms of the original momenta $p_a$, the same symmetry reads
\begin{equation}
Qp_a = -\Gamma^c{}_{ab}\theta^b p_c + \frac12\,\widetilde R_{bc}{}^d{}_a\theta^b\theta^c\chi_d, \qquad
Q\chi_a = -p_a+\Gamma^c{}_{ab}\theta^b\chi_c,
\end{equation}
where the curvature term in $Qp_a$ is again the curvature of the transposed connection.

The action~\eqref{DPSMaction} may be written as a $Q$-exact term 
\begin{equation}
\label{Psi}
S_{\mathrm{DPSM}} =Q\Psi, \qquad
\Psi = -\int_\Sigma \chi_a\wedge \left(dx^a+\frac12\,\pi^{ab}\lambda_b\right).
\end{equation}

Finally, varying the component action gives the first-order field equations
\begin{align}
dx^a+\pi^{ab}p_b&=0,\cr
\nabla\theta^a+\frac12\,\widetilde R_{cd}{}^{ab}\theta^c\theta^d\chi_b&=0,\cr
\nabla\chi_a-\frac12\,\widetilde R_{ab}{}^{cd}\theta^b\chi_c\wedge\chi_d&=0,\cr
\nabla p_a-\widetilde R_{ab}{}^c{}_d\theta^d dx^b\wedge\chi_c
+\frac14\,\nabla_a\widetilde R_{bc}{}^{de}\theta^b\theta^c\chi_d\wedge\chi_e&=0.
\end{align}
These equations form a Cartan-integrable first-order system. Their integrability conditions are equivalent to the target-space Jacobi identities above.

\subsection{Example: a differential Poisson three-torus}
We conclude this section with an example of a differential Poisson structure which is not symplectic and has nonzero transposed curvature. The target space is the three-torus
\begin{equation}
T^3
=\mathbb R^3/(2\pi\mathbb Z)^3,
\end{equation}
with angular coordinates $(x,y,\varphi)$.

The torus $T^3$ has a global frame given by $\{\partial_x,\partial_y,\partial_\varphi\}$. We shall specify the connections below globally by matrices in this frame. For a connection $\Gamma$, we write
\begin{equation}
\nabla_{\partial_a}\partial_b = \Gamma^c{}_{ab}\partial_c, \qquad
(\Gamma_a)^c{}_{b}:=\Gamma^c{}_{ab},
\end{equation}
where all indices $a,b,c$ run over $\{x,y,\varphi\}$. Since the coefficients used below are smooth periodic functions on $T^3$, they define global affine connections on the tangent bundle $TT^3$.

Take the constant rank-two Poisson tensor
\begin{equation}
\pi =\partial_x\wedge\partial_y.
\end{equation}
In the basis $\{\partial_x,\partial_y,\partial_\varphi\}$, its matrix form is
\begin{equation}
\pi^{ab}
=\begin{pmatrix}
0&1&0\\
-1&0&0\\
0&0&0
\end{pmatrix}.
\end{equation}
Thus, the Schouten bracket $[\pi,\pi]=0$, $\pi$ has constant rank two, its symplectic leaves are the two-tori $\varphi=\mathrm{constant}$, and it has no inverse on $T^3$. 

We choose the connection $\Gamma$ to be 
\begin{equation}
\Gamma_x
=\begin{pmatrix}
-1&-1&0\\
1&0&0\\
0&0&0
\end{pmatrix}, \qquad
\Gamma_y
=\begin{pmatrix}
0&-1&0\\
1&1&0\\
0&0&0
\end{pmatrix}, \qquad
\Gamma_\varphi=0.
\end{equation}
Since the connection matrices are constant, the curvature of $\Gamma$ is
\begin{equation}
R_{ab}{}^c{}_{d} = (R_{ab})^c{}_{d} , \qquad R_{ab}=[\Gamma_a,\Gamma_b].
\end{equation}
Moreover, since $\Gamma_\varphi=0$ and $[\Gamma_x,\Gamma_y]=0$, we have that
\begin{equation}
R_{xy}=R_{x\varphi}=R_{y\varphi}=0,
\end{equation}
and therefore $\Gamma$ is flat.

The transposed connection is defined by $\widetilde\Gamma^{c}{}_{ab}=\Gamma^{c}{}_{ba}$. For the choice above, this gives
\begin{equation}
\widetilde\Gamma_x
=\begin{pmatrix}
-1&0&0\\
1&1&0\\
0&0&0
\end{pmatrix}, \qquad
\widetilde\Gamma_y
= \begin{pmatrix}
-1&-1&0\\
0&1&0\\
0&0&0
\end{pmatrix},\qquad
\widetilde\Gamma_\varphi=0.
\end{equation}
These matrices preserve the Poisson tensor $\pi=\partial_x\wedge\partial_y$. Since its components are constant, the condition $\widetilde\nabla\pi=0$ reduces to
\begin{equation}
\widetilde\Gamma^a{}_{ib}\,\pi^{bc}+\widetilde\Gamma^c{}_{ib}\,\pi^{ab}=0,\qquad
i\in\{x,y,\varphi\}.
\end{equation}
For $i=\varphi$ this is trivial, while for $i=x,y$ it is satisfied by the displayed traceless matrices $\widetilde\Gamma_x$ and $\widetilde\Gamma_y$ displayed above.

The transposed curvature is nonzero. The only potentially nonzero components are those with lower indices $x,y$. Indeed, since the matrices $\widetilde \Gamma_x$ and $\widetilde \Gamma_y$ are constant, and $\widetilde \Gamma_\varphi=0$, it follows that
\begin{equation}
\widetilde R_{x\varphi} =[\widetilde\Gamma_x,\widetilde\Gamma_\varphi]=0, \qquad
\widetilde R_{y\varphi}= [\widetilde\Gamma_y,\widetilde\Gamma_\varphi]=0,
\end{equation}
and
\begin{equation}
\widetilde R_{xy}=[\widetilde\Gamma_x,\widetilde\Gamma_y]
=\begin{pmatrix}
1&2&0\\
-2&-1&0\\
0&0&0
\end{pmatrix},
\end{equation}
together with $\widetilde R_{yx}=-\widetilde R_{xy}$. 

It remains to check the Jacobi identities~\eqref{Jacobi}-\eqref{JacobiRR}. The torsion identity
\begin{equation}
\pi^{d[a}T^b{}_{de}\pi^{c]e}=0
\end{equation}
holds because $\pi$ has support only in the $x,y$ directions. A complete antisymmetrization over three upper indices must either repeat one of these two directions, or include the $\varphi$ direction. In the first case the antisymmetrization vanishes, while in the second case one of the Poisson components vanishes.

The identity
\begin{equation}
\pi^{ae}\pi^{bf}R_{ab}{}^c{}_{d}=0
\end{equation}
holds directly because $R_{ab}{}^c{}_{d}=0$. The curvature derivative identity
\begin{equation}
\pi^{ab}\nabla_b\widetilde R_{cd}{}^{ef}=0 
\end{equation}
follows from straightforward calculation. The components of $\widetilde R$ are constant, and the only
nonzero lower curvature directions are $x,y$. The covariant derivative $\nabla$ preserves this structure. Substituting the matrices above into the covariant derivative gives zero after contraction with $\pi^{ab}$.

Finally, the curvature-square identity~\eqref{JacobiRR}
\begin{equation}
\widetilde R_{a[b}{}^{(cd}\widetilde R_{mn]}{}^{p)q}=0
\end{equation}
also holds. The lower indices of $\widetilde R$ have support only in the two directions $x,y$. Since the expression antisymmetrizes over the three indices $b,m,n$, either only $x,y$ occur, in which case the antisymmetrization over three slots vanishes, or one index is $\varphi$, in which case one curvature factor is zero.

Thus the data above define a DPSM background satisfying
\begin{equation}
\widetilde\nabla_c\pi^{ab}=0, \qquad R_{ab}{}^{c}{}_d=0, \qquad \widetilde R_{ab}{}^{c}{}_d\neq 0,
\end{equation}
together with the remaining DPSM Jacobi identities. Since $\operatorname{rank}\pi=2<3$, this is a degenerate DPSM background with nonzero transposed curvature.

\section{Symplectic reduction and relation to the A-model}
\label{section3}
We now examine the nondegenerate phase of the DPSM. In this case, the Poisson tensor $\pi^{ab}$ is assumed to be invertible, with inverse
\begin{equation}
\label{omegaInverse}
\omega_{ab}:=(\pi^{-1})_{ab},\qquad \pi^{ac}\omega_{cb}=\delta^a{}_b.
\end{equation}
The two-form
\begin{equation}
\omega=\frac12\,\omega_{ab}\,dx^a\wedge dx^b
\end{equation}
is closed by the Jacobi identity for $\pi$, and hence defines a symplectic structure on $M$.

The aim of this section is to show that the symplectic reduction of the DPSM gives a
classical A-type model action on this symplectic target. The resulting theory is formulated
in symplectic terms and is not tied to a K\"ahler metric.

\subsection{Reduced A-model action and torsion-induced curvature}
We begin from the component DPSM action~\eqref{DPSMaction}. The dependence on the one-form momenta $p_a$ is purely algebraic, and their dependence is contained in
\begin{equation}
\int_\Sigma\left[p_a\wedge dx^a+\frac12\,\pi^{ab}p_a\wedge p_b\right].
\end{equation}
The equation of motion for $p_a$ is $dx^a+\pi^{ab}p_b=0$ which, using \eqref{omegaInverse}, gives
\begin{equation}
\label{piout}
p_a=-\omega_{ab}dx^b.
\end{equation}
The relation above amounts to completing the square and writing
\begin{equation}
p_a\wedge dx^a+\frac12\,\pi^{ab}p_a\wedge p_b=\frac12\,\pi^{ab}\bigl(p_a+\omega_{ac}dx^c\bigr)\wedge\bigl(p_b+\omega_{bd}dx^d\bigr)+\frac12\,\omega_{ab}dx^a\wedge dx^b.
\end{equation}
Thus, integrating out the momenta $p_a$ in~\eqref{DPSMaction} produces the reduced action
\begin{equation}
\label{Ared}
S_{\mathrm{red}}=\int_\Sigma\left[\frac12\,\omega_{ab}dx^a\wedge dx^b+\chi_a\wedge\nabla\theta^a+\frac14\,\widetilde R_{ab}{}^{cd}\theta^a\theta^b\chi_c\wedge\chi_d\right].
\end{equation}
The first term above is the topological symplectic term $\int_\Sigma X^*(\omega)$. The remaining terms are the fermionic kinetic term and a quartic curvature interaction.

The action \eqref{Ared} has the classical form of the A-model action, but its geometry is different from the usual K\"ahler presentation. The connection in the kinetic term is the flat parent connection $\Gamma$, while the curvature in the quartic interaction is that of the transposed Poisson-compatible connection $\widetilde\Gamma$. 

Recalling that, in the symplectic DPSM phase the Jacobi identities force the connection $\Gamma$ to be flat, then relation~\eqref{RT} holds, and therefore
\begin{equation}
\widetilde R_{ab}{}^{cd}=\pi^{de}\left(-2\nabla_{[a}T^c{}_{b]e}+2T^c{}_{[a|f|}T^f{}_{b]e}\right).
\end{equation}
Substituting this into~\eqref{Ared}, the reduced theory can be written as
\begin{equation}
S_{\mathrm{red}}=\int_\Sigma\left[\frac12\,\omega_{ab}dx^a\wedge dx^b+\chi_a\wedge\nabla\theta^a+\frac12\,\pi^{de}\left(-\nabla_{[a}T^c{}_{b]e}+T^c{}_{[a|f|}T^f{}_{b]e}\right)\theta^a\theta^b\chi_c\wedge\chi_d\right].
\end{equation}
Thus the quartic curvature coupling is induced by the torsion of a flat parent connection.

\subsection{Induced cohomological symmetry}
After eliminating the momenta using~\eqref{piout}, the shifted field~\eqref{lambda} becomes the composite expression
\begin{equation}
\lambda_a^{\mathrm{red}} = -\omega_{ab}dx^b - \Gamma^c{}_{ab}\theta^b\chi_c.
\end{equation}
Therefore the parent symmetry~\eqref{Qaction} induces the following transformations on the reduced fields
\begin{equation}
\label{Qred}
Qx^a=\theta^a, \qquad Q\theta^a=0, \qquad Q\chi_a = \omega_{ab}dx^b + \Gamma^c{}_{ab}\theta^b\chi_c.
\end{equation}

The induced symmetry is no longer off-shell nilpotent in general. Indeed, since $\lambda_a^{\mathrm{red}}$ is not an independent field, one finds
\begin{equation}
Q^2\chi_a=-\omega_{ab}E^b,
\end{equation}
where
\begin{equation}
E^a:=\nabla\theta^a + \frac12\,\widetilde R_{cd}{}^{ab}\theta^c\theta^d\chi_b
\end{equation}
is the reduced equation of motion obtained by varying~\eqref{Ared} with respect to $\chi_a$. The same equation of motion controls the variation of the reduced action. A direct calculation gives
\begin{equation}
QS_{\mathrm{red}} = \int_\Sigma E^a\wedge Q\chi_a .
\end{equation}
Hence the reduced transformations leave $S_{\mathrm{red}}$ invariant on shell.

The same point can be seen from the $Q$-exact form of the parent DPSM action. After reduction, the $Q$-primitive~\eqref{Psi} becomes
\begin{equation}
\Psi_{\mathrm{red}}
= -\int_\Sigma \chi_a\wedge \left( dx^a+\frac12\,\pi^{ab}\lambda_b^{\mathrm{red}} \right),
\end{equation}
which amounts to
\begin{equation}
S_{\mathrm{red}} = Q\Psi_{\mathrm{red}} - \frac12\int_\Sigma\chi_a\wedge E^a .
\end{equation}
Hence the reduced action is $Q$-exact only on shell. Equivalently, the reduced A-type model inherits its cohomological symmetry from the off-shell parent DPSM, but this symmetry becomes on shell after the momenta are eliminated.

\subsection{Relation to K\"ahler A-models}
The action \eqref{Ared} has the same formal structure as the classical A-model action. However, the geometric origin of its curvature coupling is different. In the usual K\"ahler A-model, the connection is the Levi--Civita connection of a K\"ahler metric and the quartic term is governed by the corresponding K\"ahler curvature. In the DPSM reduction, by contrast, the connection $\Gamma$ is flat by the parent Jacobi identities, while the quartic term is governed by the curvature of the transposed Poisson-compatible connection $\widetilde\Gamma$.

Thus the strict DPSM phase does not recover arbitrary K\"ahler A-model targets. If one imposed a vanishing torsion $T=0$, then $\widetilde\Gamma=\Gamma$, hence the curvatures of both connections would be equal $\widetilde R=R$. But in the nondegenerate DPSM phase $R=0$, so the quartic coupling would vanish. Nontrivial reduced curvature therefore requires nonzero torsion of the flat parent connection.

This is precisely why the resulting class is naturally non-K\"ahler. The construction is formulated in symplectic terms and does not require an integrable complex structure or a K\"ahler metric. When a compatible complex structure is present, one may decompose \eqref{Ared} in complex coordinates and compare it with the usual A-model expression. But the DPSM does not identify the curvature coupling with Levi--Civita curvature unless additional conditions are imposed.

The appropriate conclusion is therefore not that the DPSM rederives the ordinary K\"ahler A-model for arbitrary targets. Rather, it constructs a distinguished class of A-type models on symplectic targets equipped with flat parent connection data.

\section{Classical observables and the differential Poisson algebra}
\label{section4}
We now turn to the algebraic structure carried by the classical local observables. The reduced theory obtained in Section~\ref{section3} has the usual de Rham sector of A-model observables. Differential forms on the target define local operators, and the cohomological symmetry acts as the de Rham differential. The DPSM adds a second piece of structure. Since the parent target $T[1]M$ carries a graded Poisson tensor, the de Rham complex is equipped with a degree-zero differential Poisson bracket. Together with the wedge product and the de Rham differential, this makes the observable complex into a differential graded Poisson algebra.

\subsection{Local observables and the de Rham complex}
The local observables relevant for the present discussion are functions on the graded target $T[1]M$ evaluated on the worldsheet zero-form fields. Under the canonical identification~\eqref{isom}, a differential form
\begin{equation}
\eta=\frac1{p!}\eta_{a_1\cdots a_p}(x)\,dx^{a_1}\wedge\cdots\wedge dx^{a_p}
\end{equation}
corresponds to the homogeneous function
\begin{equation}
\widehat\eta(x,\theta)=\frac1{p!}\eta_{a_1\cdots a_p}(x)\theta^{a_1}\cdots\theta^{a_p}.
\end{equation}
Evaluating this function on the fields gives the local observable
\begin{equation}
\mathcal O_\eta(\sigma):=\widehat\eta\bigl(x(\sigma),\theta(\sigma)\bigr).
\end{equation}
The target cohomological vector field $Q$ on $T[1]M$ acts on $\widehat\eta$ as
\begin{equation}
Q\widehat\eta=\widehat{d\eta}.
\end{equation}
The field-space cohomological symmetry $\mathcal Q$ acts on the corresponding local observable by pullback, and hence
\begin{equation}
\mathcal Q\mathcal O_\eta=\mathcal O_{d\eta}.
\end{equation}
Thus the classical local observable complex is the de Rham complex of the target. In particular, on a closed worldsheet
\begin{equation}
H_{\mathcal Q}^\bullet(\mathrm{Obs}_{\mathrm{loc}})
\cong H^\bullet_{\mathrm{dR}}(M).
\end{equation}
This statement is independent of whether the target is K\"ahler. It only uses the graded target $T[1]M$ and the canonical cohomological vector field. Hence the non-K\"ahler A-models obtained by DPSM reduction have the same classical de Rham observable cohomology as the usual A-model. This coincidence holds at the level of $\mathcal Q$-cohomology. As we discuss next, the DPSM in addition equips the observable complex with a differential Poisson bracket that is not manifest in the usual formulation, so the two theories differ in chain-level structure even where their cohomologies agree. 

\subsection{Differential Poisson brackets of forms}
The additional structure comes from the graded Poisson tensor $\Pi$ on $T[1]M$. Its induced bracket on $C(T[1]M)$ becomes, under the identification with differential forms, a differential Poisson bracket~\cite{Arias:2015wha} on $\Omega^\bullet(M)$. In the symplectic DPSM gauge used above, this bracket is
\begin{equation}
\label{dPB}
\{\eta,\rho\}=\pi^{ab}(\nabla_a\eta)\wedge(\nabla_b\rho)+(-1)^{|\eta|}\,\widetilde R^{ab}\wedge\iota_a\eta\wedge\iota_b\rho,
\end{equation}
where $|\eta|$ is the form degree of $\eta$, $\iota_a$ denotes contraction along $\partial_a$, and
\begin{equation}
\widetilde R^{ab}:=\frac12\,\widetilde R_{cd}{}^{ab}\,dx^c\wedge dx^d.
\end{equation}
The curvature appearing here is the same curvature that appears in the quartic coupling of the reduced action.

The bracket \eqref{dPB} has degree zero, is graded antisymmetric and obeys the graded Leibniz rule, that is
\begin{equation}
|\{\eta,\rho\}|=|\eta|+|\rho|, \qquad \{\eta,\rho\}=-(-1)^{|\eta| |\rho|}\{\rho,\eta\},
\end{equation}
and
\begin{equation}
\{\eta,\rho\wedge\sigma\}=\{\eta,\rho\}\wedge\sigma+(-1)^{|\eta||\rho|}\rho\wedge\{\eta,\sigma\}.
\end{equation}
The graded Jacobi identity follows from the parent condition $[\Pi,\Pi]=0$.

Compatibility of $\Pi$ with the cohomological vector field $Q$ implies that the de Rham differential acts as a derivation of the bracket
\begin{equation}
\label{dder}
d\{\eta,\rho\}=\{d\eta,\rho\}+(-1)^{|\eta|}\{\eta,d\rho\}.
\end{equation}
Thus $\Omega^\bullet(M)$ is a differential graded Poisson algebra. Moreover, \eqref{dder} ensures that~\eqref{dPB} descends to de Rham cohomology. Indeed, the bracket of two closed forms is closed and the bracket with an exact form is exact, so it descends to a well-defined bracket on $H^\bullet_{dR}(M)$. This bracket is part of the DPSM-induced chain-level structure and is not manifest in the usual classical A-model formulation.

\subsection{Differential graded Poisson algebra and A-model observables}
We now record the homotopy-algebraic consequence of the differential Poisson structure constructed above. Since the bracket is induced by the parent graded Poisson tensor on $T[1]M$, the observable complex inherits not only the de Rham differential but also a canonical chain-level Poisson bracket.

The identities established above show that the classical observable complex $\Omega^\bullet(M)$, equipped with the de Rham differential $d$, the wedge product, and the differential Poisson bracket $\{\cdot,\cdot\}$, is the differential graded Poisson algebra
\begin{equation}
\bigl(\Omega^\bullet(M),d,\wedge,\{\cdot,\cdot\}\bigr).
\end{equation}
Indeed, the bracket has degree zero, is graded antisymmetric, satisfies the graded Jacobi identity, obeys the graded Leibniz rule with respect to the wedge product, and is compatible with the de Rham differential through~\eqref{dder}. In particular, the operations
\begin{equation}
\ell_1(\eta)=d\eta, \qquad \ell_2(\eta,\rho)=\{\eta,\rho\}, \qquad \ell_n=0,\quad n\geq 3,
\end{equation}
define a strict $L_\infty$-algebra on the same observable complex. The corresponding identities reduce to \(d^2=0\), the derivation property~\eqref{dder}, and the graded Jacobi identity of the bracket.

This chain-level structure is not manifest in the usual classical A-model formulation. There one emphasizes the de Rham cohomology ring $H^\bullet_{\mathrm{dR}}(M)$, with product induced by the wedge product of forms. After quantization on a K\"ahler target, this ring is deformed by worldsheet instantons to quantum cohomology. By contrast, the DPSM perspective keeps track of the differential Poisson bracket~\eqref{dPB} before passing to $\mathcal Q$-cohomology.

For the reduced A-type models constructed here, the differential Poisson bracket is part of the same geometric package as the action. The same transposed curvature $\widetilde R$ enters both the bracket on forms and the quartic term of the reduced action~\eqref{Ared}. Thus the reduced action, the cohomological symmetry, and the differential graded Poisson algebra of observables all come from the parent graded Poisson tensor on $T[1]M$.

\section{Examples and obstructions}
\label{section5}
Prior to concluding, we examine the geometric scope of the strict symplectic DPSM reduction at the level of some examples. The class obtained above is not the class of all A-model targets, nor is it the class of all symplectic manifolds. It consists of symplectic targets equipped with parent connection data satisfying the differential-Poisson compatibility conditions.

In the symplectic phase, the Poisson tensor is invertible, and the parent Jacobi identities~\eqref{Jacobi} imply
\begin{equation}
\label{flatness}
R_{ab}{}^c{}_d=0.
\end{equation}
Thus the parent connection $\Gamma$ is flat. The examples below illustrate two consequences of this condition, namely the global restrictions it imposes on the target geometry, and the possibility that the transposed connection $\widetilde\Gamma$ still has nonzero curvature.

\subsection{A basic obstruction}
Locally, a flat connection can always be trivialized, so condition~\eqref{flatness} does not by itself look restrictive. Globally, however, flat connections may still carry monodromy. If the target is simply connected, this possibility disappears, and flatness forces the tangent bundle to be globally trivial. This simple observation gives the following obstruction.

\begin{proposition}
Let $M$ be simply connected and compact. If $M$ admits a strict symplectic DPSM background, then $TM$ is trivializable. In particular, any simply connected symplectic manifold with nontrivial tangent bundle is excluded from the strict symplectic DPSM class.
\end{proposition}

\begin{proof}
Let $(M,\omega,\Gamma)$ be a strict symplectic DPSM background. Since $\omega$ is nondegenerate, the inverse Poisson tensor $\pi=\omega^{-1}$ is also nondegenerate, and the second Jacobi identity in~\eqref{Jacobi} implies $\Gamma$ is a flat connection on $TM$.

A flat connection has trivial parallel transport around contractible loops, or equivalently trivial restricted holonomy. Since $M$ is simply connected, every loop is contractible. Hence the full holonomy of $\Gamma$ is trivial. Parallel transport is therefore independent of the path.

Next, choose a point $p\in M$ and a basis of $T_pM$. By parallel transporting this basis to every other point of $M$, one obtains a globally defined frame of $TM$. Hence
\begin{equation}
TM\cong M\times\mathbb R^{\dim M}.
\end{equation}
Therefore $TM$ is trivializable or, equivalently, $M$ is parallelizable.
\end{proof}

The obstruction above excludes many standard K\"ahler A-model targets. For example, $\mathbb{CP}^n$ is simply connected with $\chi(\mathbb{CP}^n)=n+1$. Since a compact parallelizable manifold admits a nowhere-vanishing vector field, the Poincar\'e--Hopf theorem implies that its Euler characteristic vanishes. Therefore $T\mathbb{CP}^n$ is not trivializable, and $\mathbb{CP}^n$ cannot carry strict symplectic DPSM data.  

Similarly, a K3 surface is simply connected and $\chi(\mathrm{K3})=24$, so $T(\mathrm{K3})$ is not trivializable, and K3 is also excluded.

\subsection{Flat affine symplectic manifolds}
The simplest examples are affine symplectic manifolds. By this we mean symplectic manifolds equipped with a flat, torsion-free connection $\Gamma$ preserving the symplectic form
\begin{equation}
R(\Gamma)=0, \qquad T(\Gamma)=0, \qquad \nabla\omega=0.
\end{equation}
Since $T(\Gamma)=0$, the transposed connection agrees with the parent connection, $\widetilde\Gamma=\Gamma$, and hence $\widetilde R=R=0$. In affine coordinates, the reduced action~\eqref{Ared} becomes
\begin{equation}
S_{\mathrm{red}} = \int_\Sigma \left[\frac12\,\omega_{ab}dx^a\wedge dx^b + \chi_a\wedge d\theta^a \right].
\end{equation}
Basic flat affine examples are $\mathbb R^{2n}$ with a constant symplectic form, and $T^{2n}$ with the standard flat connection together with a translation-invariant symplectic form. Complex tori
\begin{equation}
M=\mathbb C^n/\Lambda,
\end{equation}
where $\Lambda$ is a lattice in $\mathbb C^n$, give the corresponding complex examples when equipped with the standard flat connection and a translation-invariant K\"ahler form.

\subsection{The Kodaira--Thurston manifold}
We now give a compact symplectic non-K\"ahler example in the strict DPSM class with nonzero transposed curvature. The example is the Kodaira--Thurston manifold~\cite{Kodaira1963,Thurston1976}, a four-dimensional nilmanifold which
may be viewed as a nontrivial $T^2$-bundle over $T^2$. It is the first known instance of a compact manifold admitting both complex and symplectic structures without a compatible K\"ahler structure (see~\cite{Bazzoni2014} for a modern account and generalizations). 
Its non-K\"ahler character follows from the general rigidity of K\"ahler nilmanifolds, which are necessarily diffeomorphic to a torus~\cite{BensonGordon1988,Hasegawa1989}, so the odd first Betti number $b_1=3$ of the Kodaira--Thurston manifold already excludes it.

Recall that a nilmanifold is a compact quotient $M=G/\Lambda$ of a simply connected nilpotent Lie group $G$ by a lattice $\Lambda$. Differential forms on $G$ which are invariant under left translations descend to globally defined forms on $M$.

The Kodaira--Thurston manifold $\mathrm{KT}^4$ is a four-dimensional nilmanifold admitting a global invariant coframe $\{e^1,e^2,e^3,e^4\}$ satisfying
\begin{equation}
de^1=de^2=de^3=0, \qquad de^4=e^1\wedge e^2.
\end{equation}
Let $\{E_1,E_2,E_3,E_4\}$ be the dual invariant frame. The only nonzero Lie bracket is
\begin{equation}
[E_1,E_2]=-E_4,
\end{equation}
together with its antisymmetric counterpart. Equivalently, if
\begin{equation}
\label{EEbra}
[E_a,E_b]=C^c{}_{ab}E_c,
\end{equation}
then
\begin{equation}
C^4{}_{12}=-1,
\qquad
C^4{}_{21}=1
\end{equation}
are the only nonzero structure constants.

Consider the two-form
\begin{equation}
\omega=e^1\wedge e^4+e^2\wedge e^3.
\end{equation}
In the coframe $\{e^1,e^2,e^3,e^4\}$, its matrix form is
\begin{equation}
\omega_{ab}=\begin{pmatrix}
0&0&0&1\\
0&0&1&0\\
0&-1&0&0\\
-1&0&0&0
\end{pmatrix}.
\end{equation}
$\omega$ is closed, since
\begin{equation}
d\omega = d(e^1\wedge e^4)+d(e^2\wedge e^3)
=-e^1\wedge de^4
=-e^1\wedge e^1\wedge e^2
=0,
\end{equation}
and nondegenerate because
\begin{equation}
\omega^2 = 2\,e^1\wedge e^4\wedge e^2\wedge e^3 \neq 0.
\end{equation}
Thus $\mathrm{KT}^4$ is symplectic. It is compact and non-K\"ahler, since $b_1(\mathrm{KT}^4)=3$, whereas Hodge symmetry implies that every odd-degree Betti number of a compact K\"ahler manifold is even.

The inverse Poisson tensor is
\begin{equation}
\pi=-E_1\wedge E_4-E_2\wedge E_3.
\end{equation}
Equivalently, in the frame $\{E_1,E_2,E_3,E_4\}$, its matrix is
\begin{equation}
\pi^{ab}=\begin{pmatrix}
0&0&0&-1\\
0&0&-1&0\\
0&1&0&0\\
1&0&0&0
\end{pmatrix}.
\end{equation}

We now choose the parent connection $\Gamma$ which, in the invariant frame, is defined by
\begin{equation}
\nabla_{E_a}E_b=\Gamma^c{}_{ab}E_c, \qquad (\Gamma_a)^c{}_{b}:=\Gamma^c{}_{ab}.
\end{equation}
Take
\begin{equation}
\Gamma_1=\begin{pmatrix}
0&0&0&0\\
0&0&0&0\\
-1&0&0&0\\
0&0&0&0
\end{pmatrix},\qquad
\Gamma_2=
\begin{pmatrix}
0&0&0&0\\
0&0&0&0\\
0&0&0&1\\
0&0&0&0
\end{pmatrix}, \qquad
\Gamma_3=0,\qquad \Gamma_4=0.
\end{equation}

Since the invariant frame is not a coordinate frame, the Lie bracket~\eqref{EEbra} of the frame vectors contributes to the curvature. In this frame
\begin{equation}
R(E_a,E_b) = [\Gamma_a,\Gamma_b]-C^c{}_{ab}\Gamma_c.
\end{equation}
The only nonzero structure constants are $C^4{}_{12}=-1$ and $C^4{}_{21}=1$, but $\Gamma_4=0$. Moreover, the displayed matrices $\Gamma$ commute. Hence
\begin{equation}
R(\Gamma)=0.
\end{equation}
Thus the parent connection is flat.

The torsion is
\begin{equation}
T^c{}_{ab} = 2\, \Gamma^c{}_{[ab]} - C^c{}_{ab}.
\end{equation}
For the connection above, the nonzero torsion components are
\begin{equation}
T^3{}_{24}=1,
\qquad
T^3{}_{42}=-1,
\qquad
T^4{}_{12}=1,
\qquad
T^4{}_{21}=-1.
\end{equation}
Hence the parent connection is flat but has nonzero torsion.

In a non-coordinate frame, the transposed connection~$\widetilde\nabla_{E_a}E_b := \nabla_{E_b}E_a + [E_a,E_b]$. Expanding with~\eqref{EEbra}, one finds
\begin{equation}
(\widetilde\Gamma_a)^c{}_{b} = \Gamma^c{}_{ba} + C^c{}_{ab}.
\end{equation}
In a coordinate frame $C^c{}_{ab}=0$ and this reduces to $\widetilde\Gamma^c{}_{ab}=\Gamma^c{}_{ba}$.
For the connection above, one obtains
\begin{equation}
\widetilde\Gamma_1=\begin{pmatrix}
0&0&0&0\\
0&0&0&0\\
-1&0&0&0\\
0&-1&0&0
\end{pmatrix},\qquad
\widetilde\Gamma_2=
\begin{pmatrix}
0&0&0&0\\
0&0&0&0\\
0&0&0&0\\
1&0&0&0
\end{pmatrix},
\end{equation}
and
\begin{equation}
\widetilde\Gamma_3=0,\qquad
\widetilde\Gamma_4=
\begin{pmatrix}
0&0&0&0\\
0&0&0&0\\
0&1&0&0\\
0&0&0&0
\end{pmatrix}.
\end{equation}

The transposed connection is Poisson-compatible. Since $\pi=\omega^{-1}$, this is equivalent to preserving $\omega$. In matrix form, one checks the algebraic identities
\begin{equation}
\widetilde\Gamma_i^T\omega+\omega\widetilde\Gamma_i=0,
\qquad
i=1,\ldots,4,
\end{equation}
for the displayed matrices. Therefore $\widetilde\nabla\pi=0$.

As for the curvature $\widetilde R$ of the transposed connection $\widetilde \Gamma$, this is nonzero. In the same invariant frame
\begin{equation}
\widetilde R(E_a,E_b) = [\widetilde\Gamma_a,\widetilde\Gamma_b] - C^c{}_{ab}\widetilde\Gamma_c.
\end{equation}
The only nonzero curvature direction is $(E_1,E_2)$, together with antisymmetry in the lower pair. A direct computation gives
\begin{equation}
\widetilde R(E_1,E_2)
=
\begin{pmatrix}
0&0&0&0\\
0&0&0&0\\
0&1&0&0\\
0&0&0&0
\end{pmatrix}.
\end{equation}
Equivalently
\begin{equation}
\widetilde R_{12}{}^3{}_{2}=1.
\end{equation}
With the convention
\begin{equation}
\widetilde R_{ab}{}^{cd} := \pi^{de}\widetilde R_{ab}{}^c{}_{e},
\end{equation}
it follows that
\begin{equation}
\widetilde R_{12}{}^{33}=1,
\qquad
\widetilde R_{21}{}^{33}=-1.
\end{equation}
Therefore, the transposed curvature $\widetilde R (\widetilde \Gamma)$ is nonzero.

It remains to check the symplectic limit of the parent Jacobi identities~\eqref{Jacobi}-\eqref{JacobiRR}. The torsion identity
\begin{equation}
\pi^{d[a}T^b{}_{de}\pi^{c]e}=0
\end{equation}
is checked directly from the displayed torsion components and the matrix of $\pi$. The flatness condition
$R_{ab}{}^{c}{}_d=0$ has already been verified. The curvature derivative identity
\begin{equation}
\pi^{ab}\nabla_b\widetilde R_{cd}{}^{ef}=0
\end{equation}
also holds by direct substitution, where the covariant derivative is taken with respect to the parent connection $\Gamma$. Finally, the curvature-square identity
\begin{equation}
\widetilde R_{a[b}{}^{(cd}\widetilde R_{mn]}{}^{p)q}=0
\end{equation}
holds because the only nonzero components are $\widetilde R_{12}{}^{33}=1=-\widetilde R_{21}{}^{33}$. Thus, antisymmetrization in the lower indices forces the quadratic expression to vanish.

Thus KT${}^4$ with the data above defines a strict symplectic DPSM background with nonzero transposed curvature.

\section{Conclusions}
\label{section6}
We have shown that the symplectic phase of the DPSM gives a first-order construction of a restricted class of A-type sigma models on symplectic targets. After integrating out the one-form momenta, the reduced action takes the form
\begin{equation}
S_{\mathrm{red}} = \int_\Sigma \left[ \frac12\,\omega_{ab}dx^a\wedge dx^b
+\chi_a\wedge\nabla\theta^a
+ \frac14\,\widetilde R_{ab}{}^{cd}\theta^a\theta^b\chi_c\wedge\chi_d \right].
\end{equation}
Thus the reduced theory contains the standard ingredients of a classical A-type model, namely a topological symplectic term, a fermionic kinetic term, a quartic curvature interaction, and a cohomological symmetry. The geometric origin of these terms, however, is different from the usual K\"ahler construction. The input data are a symplectic form together with a parent connection satisfying the DPSM compatibility conditions; in particular, 
\begin{equation}
R(\Gamma)=0.
\end{equation}
The parent connection is therefore flat, but it need not be torsion-free. However, the curvature appearing in the reduced action is not the curvature of $\Gamma$. Instead, it is the curvature of the transposed Poisson-compatible connection $\widetilde\Gamma$, viz., 
\begin{equation}
\widetilde R_{ab}{}^c{}_{d} = -2\nabla_{[a}T^c{}_{b]d} + 2T^c{}_{[a|e|}T^e{}_{b]d}.
\end{equation}
This is the basic geometric distinction between the reduced DPSM model and the usual K\"ahler presentation.

The resulting model is a non-trivial restriction of the A-model.  In particular, if the target is simply connected and compact, the flatness of $\Gamma$ forces the tangent bundle to be trivializable. Hence simply connected symplectic manifolds with nontrivial tangent bundle, such as $\mathbb{CP}^n$ and K3 surfaces, cannot arise in the strict symplectic DPSM class. At the same time, the class is not empty. It contains the flat affine symplectic examples, and it also contains compact non-K\"ahler examples with nonzero transposed curvature, as illustrated by the Kodaira--Thurston background.

We have also identified the classical local observables inherited from the parent DPSM in vanishing worldsheet form-degree. Target differential forms $\eta\in\Omega^\bullet(M)$ define $\mathcal O_\eta$ satisfying $\mathcal Q\mathcal O_\eta=\mathcal O_{d\eta}$. Thus the space of classical observables contains a subspace isomorphic to the de Rham cohomology of the target. In addition, upon the reduction, the graded parent Poisson tensor induces a degree-zero differential Poisson bracket on $\Omega^\bullet(M)$. Together with the de Rham differential and wedge product, this gives the differential graded Poisson algebra described above. Equivalently, the differential $d$ and the bracket $\{\cdot,\cdot\}$ define the associated strict $L_\infty$-algebra on the same observable complex. The parent DPSM therefore determines both the reduced action and a chain-level algebraic structure on the observable complex.

The current analysis of the classical model serves as a preamble for the construction of quantum theories based on the DPSM. Indeed, the classical structure suggests a natural deformation problem. In the ordinary PSM, perturbative quantization on the disk gives a field-theoretic realization of Kontsevich's deformation quantization~\cite{Cattaneo:1999fm,
Kontsevich2003} of the Poisson algebra
\begin{equation}
\big(C(M),\{\cdot,\cdot\}_\pi\big).
\end{equation}
For the DPSM, the corresponding classical object is the differential graded Poisson algebra
\begin{equation}
\bigl(\Omega^\bullet(M),d,\wedge,\{\cdot,\cdot\}\bigr).
\end{equation}
Here the classical structure offers a suggestive starting point.
Since the DPSM action is the superfield lift of the PSM action---obtained by replacing the Poisson manifold $(M,\pi)$ with the graded target $(T[1]M,\Pi)$ in the first-order form---one might expect its disk quantization to admit an organization along the lines of the Cattaneo--Felder scheme. 
Should this be the case, it would be natural to ask whether the resulting deformation carries the differential graded Poisson algebra $\bigl(\Omega^\bullet(M),d,\wedge,\{\cdot,\cdot\}\bigr)$ into a non-strict homotopy algebra on differential forms, generalizing the deformation quantization of $\bigl(C(M),\{\cdot,\cdot\}_\pi\bigr)$. 
In the symplectic phase studied here, such a deformation would presumably be controlled by the same connection data that enter the reduced A-type action, in particular the transposed curvature $\widetilde R$. We hope to return to these questions elsewhere.

In summary, the strict symplectic DPSM reduction does not reproduce the full class of K\"ahler A-models, nor does it cover all symplectic targets. It selects a narrower differential-Poisson class, namely symplectic manifolds equipped with a flat parent connection whose transpose is Poisson-compatible and whose torsion satisfies the remaining DPSM Jacobi
identities.

\subsection*{Acknowledgments} 

We thank R. Aros, F. Caro, F. Diaz, C. Iazeolla, D. Rovere, F. Silva, D. Tempo, O. Valdivia, M. Valenzuela, and B. Vallilo for useful discussions.
The work of P.S. is supported by FONDECYT Regular grants N°1250672 and N°1252053; the MATH-AmSud project SGP 24-MATH-12 funded by ANID and Minist\`ere de l'Europe et des Affaires \'Etrang\`eres (MEAE); and UNAP-VRII Consolida grant ``Higher-spin inspired IR modifications of 3d gravity”.  
P.S. expresses his gratitude to the DCF of UNAB in Santiago and CECs in Valdivia for hospitality during various stages of this work.

\bibliographystyle{JHEP}
\bibliography{DPSMRefs}

\end{document}